\documentclass[twocolumn]{autart}    

\usepackage{color}
\usepackage{amsmath,amssymb,amsfonts}
\usepackage[bb=boondox]{mathalfa}
\usepackage{stmaryrd}
\usepackage{graphicx}          

\allowdisplaybreaks[4]

\newcommand{\0}{\mathbb{0}}
\newcommand{\I}{\mathbb{I}}
\newcommand{\1} {\mathbb{1}}

\newcommand{\bR} { {\mathbb R}}
\newcommand{\bZ} { {\mathbb Z}}

\newcommand{\id}{\mathsf{id}}

\usepackage{enumitem}
\setlist[enumerate]{itemsep=0mm,topsep=5pt}
\renewcommand{\theenumi}{(\roman{enumi})}

\newcommand{\lp}[1] { \mathsf{Lip}\left(#1\right) }
\newcommand{\olp}[1] { \mathsf{osLip} \left(#1\right) }

\newcommand{\lpP}[1] { \mathsf{Lip}_{2,P^{1/2}}(#1) }
\newcommand{\lpi}[1] { \mathsf{Lip}_i(#1) }
\newcommand{\olpP}[1] { \mathsf{osLip}_{2,P^{1/2}}(#1) }
\newcommand{\olpDP}[1] { \mathsf{osLip}_{2,([d] \otimes P)^{1/2}}(#1) }

\newcommand{\lpeta}[1] { \mathsf{Lip}_{1,[\eta]}(#1) }
\newcommand{\olpeta}[1] { \mathsf{osLip}_{1,[\eta]}(#1) }
\newcommand{\lpetaa}[1] { \mathsf{Lip}_{1,[\1_s \otimes \eta]}(#1) }
\newcommand{\olpetaa}[1] { \mathsf{osLip}_{1,[\1_s \otimes \eta]}(#1) }

\newcommand{\lpetai}[1] { \mathsf{Lip}_{\infty,[\eta]^{-1}}(#1) }
\newcommand{\olpetai}[1] { \mathsf{osLip}_{\infty,[\eta]^{-1}}(#1) }

\newcommand{\llangle}{\langle\!\langle}
\newcommand{\rrangle}{\rangle\!\rangle}

\usepackage{xcolor}

\definecolor{fblightblue}{RGB}{220,235,250}

\definecolor{yklightred}{RGB}{250,235,220}

\def\red{\hfill $\lhd$}

\definecolor{gnred}{RGB}{255,91,89}

\definecolor{gnred1}{RGB}{71,0,0} 
\definecolor{gnred2}{RGB}{117,0,0} 
\definecolor{gnred3}{RGB}{164,0,0} 
\definecolor{gnred4}{RGB}{211,0,0} 
\definecolor{gnred5}{RGB}{255,0,0} 
\definecolor{gnred6}{RGB}{255,42,34} 
\definecolor{gnred7}{RGB}{255,91,89} 

\definecolor{gnblue1}{RGB}{0,36,71}   
\definecolor{gnblue2}{RGB}{0,60,118}  
\definecolor{gnblue3}{RGB}{0,85,164}  
\definecolor{gnblue4}{RGB}{0,108,212} 
\definecolor{gnblue5}{RGB}{0,133,255}  
\definecolor{gnblue6}{RGB}{35,156,255} 
\definecolor{gnblue7}{RGB}{88,177,255} 

\definecolor{gnbrown1}{RGB}{71,27,0}  
\definecolor{gnbrown2}{RGB}{117,45,0} 
\definecolor{gnbrown3}{RGB}{164,62,0} 
\definecolor{gnbrown4}{RGB}{211,80,0} 
\definecolor{gnbrown5}{RGB}{255,97,0} 
\definecolor{gnbrown6}{RGB}{255,127,26} 
\definecolor{gnbrown7}{RGB}{255,155,86} 

\begin{document}

\begin{frontmatter}

\title{Contractivity of Multi-Stage Runge--Kutta Dynamics\thanksref{footnoteinfo}} 

\thanks[footnoteinfo]{
This paper was not presented at any IFAC 
meeting. Corresponding author Y.~Kawano. Tel. +81-82-424-7532. 
Fax +81-82-424-7193.}

\author[JP]{Yu Kawano}\ead{ykawano@hiroshima-u.ac.jp} \quad    
\author[USA]{Francesco Bullo}\ead{bullo@ucsb.edu}               

\address[JP]{Graduate School of Advanced Science and Engineering, Hiroshima University, Higashi-Hiroshima 739-8527, Japan}  
\address[USA]{Department of Mechanical Engineering and the Center for Control, Dynamical Systems, and Computation, University of California, Santa Barbara, USA}             


\begin{abstract}                          
  Many control, optimization, and learning algorithms rely on
  discretizations of continuous-time contracting systems, where
  preservation of contractivity under numerical integration is key
  for stability, robustness, and reliable fixed-point computation.
  In this paper, we establish conditions under which multi-stage
  Runge–Kutta methods preserve \emph{strong contractivity} when
  discretizing infinitesimally contractive continuous-time systems. For
  explicit Runge–Kutta methods, preservation conditions are derived by
  bounding Lipschitz constants of the associated composite stage mappings,
  leading to coefficient-dependent criteria. For implicit methods, the
  algebraic structure of the stage equations enables explicit conditions on
  the Runge–Kutta coefficients that guarantee preservation of strong
  contractivity. In the implicit case, these results extend classical
  guarantees, typically limited to weak contractivity in the Euclidean
  metric, to strong contractivity with respect to the $\ell_1$-, $\ell_2$-,
  and $\ell_\infty$-norms. In addition, we study well-definedness of
  implicit methods through an auxiliary continuous-time system associated
  with the stage equations. We show that strong infinitesimal contractivity
  of this auxiliary system is sufficient to guarantee unique solvability of
  the stage equations. This analysis generalizes standard well-definedness
  conditions and provides a dynamic implementation approach that avoids
  direct solution of the implicit algebraic equations.
\end{abstract}

\end{frontmatter}

\section{Introduction}

\emph{Problem Description and Motivation:} Numerical integration of
continuous-time dynamical systems is a central component of modern control,
optimization, and learning. Beyond accuracy, discretization schemes are
often required to preserve qualitative properties of the underlying
differential equation. In particular, infinitesimal contractivity implies
incremental stability, robustness, and exponential convergence, properties
that are fundamental in feedback design, observer synthesis, and
large-scale interconnected systems. When a contracting continuous-time
system is discretized, a natural question arises: under what conditions do
multi-stage Runge–Kutta methods preserve strong infinitesimal
contractivity?

This question is not only of theoretical interest. Numerical integration
routines are increasingly embedded in control-oriented computations,
including numerical construction of Lyapunov functions~\cite{PG-SH:24} and
contraction metrics~\cite{PG-SH-IM:23}, as well as in neural
ODEs~\cite{NG-ADM-AS-FT:25}, implicit
models~\cite{MR-RW-IRM:24,CLG-LF-GFT:25}, and inference for predictable
dynamics~\cite{XG-LK-DMZ-KLC-SWL:25}. In these contexts, reliable
fixed-point computation and stability of discrete-time dynamics rely on
preservation of strong contractivity under discretization. The objective of
this paper is to revisit the analysis of multi-stage Runge–Kutta methods
adopting methods and results from modern contraction theory, focusing on
(i) well-posedness and computational realization of the induced
discrete-time dynamics, and (ii) preservation of strong infinitesimal
contractivity.

\emph{Literature Review:} The classical notion of B-stability, introduced
in~\cite{BJC-75}, formalizes preservation of weak infinitesimal
contractivity, or non-expansiveness, under time discretization. The
analysis relies on one-sided Lipschitz constants~\cite{GD:58}, which later
became central in the study of Runge–Kutta methods and their well-posedness
properties; these foundational developments are documented in the
influential monographs~\cite{HE-NSP-WG:93,HE-WG:96}. Related notions, such
as asymptotic B-stability~\cite{HE-ZM:96}, address asymptotic
convergence. Classical results on well-posedness provide sufficient
conditions in terms of Lipschitz constants with respect to arbitrary
norms~\cite{HE-NSP-WG:93} or one-sided Lipschitz constants with respect to
the $\ell_2$-norm~\cite{HE-WG:96}. However, the focus is primarily on weak
contractivity with respect to the $\ell_2$-norm.

In parallel, contraction theory was introduced to the systems and control
community in~\cite{LW-SJJE:98}, providing a differential framework for
analyzing exponential convergence of nonlinear systems. Although
contraction theory and numerical integration share common tools and
concepts, including one-sided Lipschitz constants and infinitesimal
contractivity, a direct theoretical connection between contraction theory
and the analysis of general multi-stage Runge–Kutta methods has not been
systematically developed.

More recently, using terminology from monotone operator theory,
preservation of strong infinitesimal contractivity has been established for
the forward Euler and implicit Euler methods~\cite{AD-SJ-AVP-FB:24}, which
are special cases of multi-stage Runge–Kutta schemes. These results rely on
standard and one-sided Lipschitz constants with respect to arbitrary
norms. To the best of our knowledge, preservation of strong infinitesimal
contractivity for general explicit and implicit multi-stage Runge–Kutta
methods has remained largely unexplored.

\emph{Contributions:} The contributions of this paper address both
well-posedness and preservation of strong infinitesimal contractivity for
multi-stage Runge–Kutta methods, and contain all elements outlined in the
original formulation.

First, we revisit the classical problem of well-posedness. When applied to
continuous-time systems, multi-stage Runge–Kutta methods generate implicit
discrete-time systems, referred to as multi-stage Runge–Kutta dynamics. We
introduce an auxiliary dynamics associated with a given Runge–Kutta scheme
and show that strong infinitesimal contractivity of this auxiliary system,
with respect to an arbitrary norm, guarantees that the corresponding
multi-stage Runge–Kutta dynamics admit a unique explicit form. The
resulting condition is simple and contains the classical well-posedness
conditions based on Lipschitz constants and one-sided Lipschitz
constants~\cite{HE-NSP-WG:93,HE-WG:96} as special cases. Moreover, the
auxiliary dynamics provides a systematic framework for implementing
multi-stage Runge–Kutta methods. In particular, when the auxiliary system
is strongly infinitesimally contractive and has a finite Lipschitz
constant, its forward Euler discretization with sufficiently small step
size yields a practical implementation of the multi-stage Runge–Kutta
method. This construction relies on the fact that forward Euler with
sufficiently small step size preserves strong infinitesimal
contractivity~\cite{AD-SJ-AVP-FB:24}.

Second, we study preservation of strong infinitesimal contractivity under
discretization. With respect to the $\ell_2$-norm, we show that the
classical B-stability condition~\cite{HE-WG:96}, together with the
assumption that the vector field has a finite Lipschitz constant,
guarantees that multi-stage Runge–Kutta methods preserve strong
infinitesimal contractivity, that is, exponential convergence. Thus,
classical B-stability, originally formulated for weak contractivity,
combined with a bounded Lipschitz condition, ensures preservation of strong
contractivity in the $\ell_2$ setting.

Third, with respect to the $\ell_1$- and $\ell_\infty$-norms, we derive
novel sufficient conditions under which multi-stage Runge–Kutta methods
preserve strong infinitesimal contractivity. These results extend
contractivity-preserving guarantees beyond the $\ell_2$ framework of
classical B-stability and establish new norm-dependent criteria for
exponential convergence of the discrete-time dynamics.

{\it Organization:} The remainder of this paper is organized as follows. In Section~\ref{sec:pre}, we present the necessary mathematical preliminaries. In Section~\ref{sec:erk}, we estimate Lipschitz constants of explicit multi-stage Runge–Kutta methods. In Section~\ref{sec:imp}, we study the well-definedness and implementation of multi-stage Runge–Kutta methods through the strong infinitesimal contractivity of the corresponding auxiliary systems. Also, we provide conditions under which a multi-stage Runge–Kutta method preserves the strong infinitesimal contractivity of a continuous-time system. Section~\ref{sec:con} concludes this paper.

\emph{Notation:}
The ring of integers, field of real numbers, and set of positive real numbers are denoted by~$\bZ$,~$\bR$, and~$\bR_{>0}$, respectively. 
The identity mapping on~$\bR^n$ is denoted by~$\id_n$.
The~$n$-component vector and~$n \times m$-matrix whose all components are~$0$ are denoted by~$\0_n$ and~$\0_{n \times m}$, respectively. The~$n$-component vector whose all components are~$1$ is denoted by~$\1_n$. The~$n \times n$ identity matrix is denoted by~$\I_n$. For~$\eta \in \bR^n$, the diagonal matrix whose diagonal entries are equal to~$\eta$ is denoted by~$[\eta] \in \bR^{n \times n}$. The matrix obtained by taking the element-wise absolute value of~$A \in \bR^{n \times m}$ is denoted by~$|A|$. 
The Kronecker product of~$A \in \bR^{n \times m}$ and~$B \in \bR^{q \times r}$ is denoted by~$A \otimes B$ whose size is~$nq \times mr$. For~$v,w \in \bR^n$, relation~$v \ge w$ means~$v_i \ge w_i$ for all~$i=1,\dots,n$. For~$P,Q \in \bR^{n \times n}$, relation~$P \succ Q$ (resp.~$P \succeq Q$) means that~$P-Q$ is symmetric and positive (resp. semi) definite.
Given~$P \succ \0_{n \times n}$ and~$\eta \in \bR_{>0}^n$, the weighed~$\ell_2$,~$\ell_1$-, and~$\ell_\infty$-norms are denoted by~$\| x\|_{2,P^{1/2}} := \sqrt{x^\top P x}$, $\|x\|_{1,[\eta]} := \eta^\top |x|$, and~$\|x\|_{\infty,[\eta]^{-1}} := \max_{i \in \{1,\dots,n\}} |x_i|/\eta_i$, respectively.
If~$P = \I_n$ and~$\eta = \1_n$, they are simply denoted by~$\| \cdot \|_2$,~$\| \cdot \|_1$, and~$\| \cdot \|_\infty$, respectively. Given vector norm~$\| \cdot \|$, a compatible weak pairing~\cite[Definition 2.27]{FB:26} is denote by~$\llbracket \cdot ; \cdot \rrbracket$. A compatible weak pairing of~$\| \cdot \|_{2,P^{1/2}}$ is the inner product~$\llbracket x ; y \rrbracket_{2,P^{1/2}} = \llangle x, y\rrangle_{2,P^{1/2}} := y^\top P x$.
Compatible weak pairings of~${\| \cdot \|_{1,[\eta]}}$ and~$\| \cdot \|_{\infty,[\eta]^{-1}}$ are~$\llbracket x ; y \rrbracket_{1,[\eta]} = \| y\|_{1,[\eta]} {\rm sign}(y)^\top [\eta] x$ and~$\llbracket x ; y \rrbracket_{\infty,[\eta]^{-1}} = \max_{i \in I_\infty ([\eta]^{-1}y)} \eta_i^{-2} y_i x_i$,
where~$I_\infty (x) := \{ i \in \{1,\dots,n\} : |x_i| = \|x\|_\infty\}$. Given vector norm~$\| \cdot \|$, the corresponding induced matrix norm and induced log norm of~$A \in \bR^{n \times n}$ are denoted by~$\| A \| := \max_{x\in\bR^n, \|x\|=1} \|A x\|$ and~$\mu (A) := \lim_{h \to 0^+} \frac{\| \I_n + h A\|}{h}$, respectively.
Also, the corresponding smallest Lipshitz constant and smallest one-sided Lipshitz constant of~$f:\bR \times \bR^n \to \bR^n$ with respect to~$x$ are respectively denoted by
\begin{align*}
    \lp{f} 
    &:= \sup_{\substack{x \neq y \\ (t,x,y) \in \bR \times \bR^n \times \bR^n}} 
    \frac{\|f(t,x) - f(t,y)\|}{\|x-y\|}, \\
    \olp{f} 
    &:= \sup_{\substack{x \neq y \\ (t,x,y) \in \bR \times \bR^n \times \bR^n}} 
    \frac{\llbracket f(t,x) - f(t,y) ; x-y \rrbracket}{\|x-y\|^2}.
\end{align*}



\section{Problem Setup}\label{sec:pre}
Consider a continuous-time nonlinear system:
\begin{align}\label{eq:sys}
    \dot x(t) = f(t, x(t)),
\end{align}
where~$f:\bR \times \bR^n \to \bR^n$ is continuous. 

We recall~\cite[Definitions 3.8 and 4.2]{FB:26} of contractivity.

\begin{defn}
\textbf{\textup{(Contractivity in Continuous-Time)}}
    Let~$\| \cdot \|$ be a norm on~$\bR^n$ with a compatible weak pairing, and let~$f:\bR \times \bR^n \to \bR^n$ be continuous. 
    The continuous-time system~\eqref{eq:sys} is
    \begin{enumerate}[nosep]
    \item \emph{weakly infinitesimally contracting} with respect to~${\|
      \cdot \|}$ if~$\olp{f} \le 0$, and
    \item \emph{strongly infinitesimally contracting} with respect to~$\|
      \cdot \|$ if there exists~$\lambda \in \bR_{>0}$ such that~$\olp{f}
      \le - \lambda$, where~$\lambda$ is called a \emph{contractivity
      rate}.~\red
    \end{enumerate}
\end{defn}

We note that weak infinitesimal contractivity implies the uniqueness of the solution~$x(t)$ to the continuous-time system~\eqref{eq:sys} for any finite~$t \ge t_0$ under any initial condition~$(t_0, x(t_0)) \in \bR \times \bR^n$~\cite[Remark 3.6]{FB:26}. Moreover, strong infinitesimal contractivity implies that an equilibrium of the system is globally exponentially stable ~\cite[Theorem 3.9]{FB:26}.

Applying an~$s$-stage Runge--Kutta method, e.g.,~\cite[eq. (12.3)]{HE-WG:96} to the continuous-time system~\eqref{eq:sys} yields the following discrete-time system:
\begin{subequations}\label{eq:dsys}
    \begin{align}
        x_{k+1} &= x_k + h (b \otimes \I_n)^\top F_c (t_k, y), \quad k \in \bZ,
        \label{eq1:dsys}\\
        y &= (\1_s \otimes x_k) + h (A \otimes \I_n) F_c (t_k, y),
        \label{eq2:dsys} \\
        &F_c (t, y) 
        :=
        \begin{bmatrix}
        f (t + c_1, y_1 ) \\ \vdots \\ f (t + c_s, y_s )
        \end{bmatrix}
        \label{eq3:dsys}
    \end{align}
\end{subequations}
with~$(t_k, x_k) \in \bR \times \bR^n$, where~$h \in \bR_{>0}$ denotes the
step size, and~$A \in \bR^{s \times s}$ and~$b, c \in \bR^s$, and~$y :=
[\begin{matrix} y_1^\top & \cdots & y_s^\top \end{matrix}]^\top \in
\bR^{sn}$ denotes the vector of stage values.  Equation~\eqref{eq1:dsys} is
called the \emph{update equation}, while \eqref{eq2:dsys}–\eqref{eq3:dsys}
are called the \emph{stage equations}.  We refer to~\eqref{eq:dsys} as the
\emph{$s$-stage Runge–Kutta dynamics}.  If~$a_{i,j} = 0$ for all $i \le j$,
then the $s$-stage Runge–Kutta method/dynamics is called \emph{explicit};
otherwise, it is called \emph{implicit}.

Our main interest in this paper is to study when a multi-stage Runge--Kutta method preserves the strong infinitesimal contractivity of a continuous-time system. To this end, we recall~\cite[Definition 1.5]{FB:26} of contractivity in discrete-time. 

\begin{defn}
  \textbf{\textup{(Contractivity in Discrete-Time)}} Let~$\| \cdot \|$ be a
  norm on~$\bR^n$, and let~$g:\bR \times \bR^n \to \bR^n$.  The
  discrete-time system:~$x_{k+1} = g (t_k, x_k)$,~$k \in \bZ$ is
  \begin{enumerate}
  \item \emph{weakly contracting} with respect to~$\| \cdot \|$
    if~$\lp{g} \le 0$, and
  \item \emph{strongly contracting} with respect to~$\| \cdot \|$ if there
    exists~$\rho \in [0,1)$ such that~$\lp{f} \le \rho$, where~$\rho$ is
      called a \emph{contractivity factor}.~\red
  \end{enumerate}
\end{defn}

\section{Explicit Runge-Kutta Methods}\label{sec:erk}

\subsection{Lipschitz Constant Estimation}
In the explicit case, we can estimate a Lipschitz constant of an~$s$-stage Runge--Kutta method as follows.

\begin{thm}\label{thm:ex}
\textbf{\textup{(Lipschitz Constant: Explicit Case)}}
    Let~$\| \cdot \|$ be a norm on~$\bR^n$. Let~$f:\bR \times \bR^n \to \bR^n$ be continuous, and let~$A \in \bR^{s \times s}$ be lower triangular with zeros on the diagonal. Define~$d_0 := \sum_{i=1}^s b_i$ and~$d_i := \sum_{j<i} a_{i,j}$. Then, a Lipschitz constant~$\rho \in \bR_{>0}$ of the explicit $s$-stage Runge--Kutta dynamics~\eqref{eq:dsys} with respect to norm~$\| \cdot \|$ is given by
    \begin{subequations}\label{eq:ex}
        \begin{align}\label{eq1:ex}
        \rho
        &:= \sum_{i=1}^s \left|\frac{b_i}{d_0}\right| 
        \lp{\id + h d_0 f} \rho_i,
        \end{align}
        where~$\rho_i$ is defined recursively by~$\rho_1 = 1$ and
        \begin{align}\label{eq2:ex}
        \rho_i
        &:= \sum_{j<i} \left|\frac{a_{i,j}}{d_i}\right| \lp{\id + h d_i f} \rho_j
        \nonumber\\
        &\quad + \lp{f} \sum_{j<i} \sum_{l<j} \left|\frac{ha_{i,j}a_{j,l}}{d_i}\right|  \rho_l,
        \quad i > 1.
        \end{align}        
    \end{subequations}
\end{thm}
\begin{pf}
    We rewrite~\eqref{eq1:dsys} and~\eqref{eq2:dsys} as
    \begin{subequations}
        \begin{align}\label{pf1-1:ex}
            x_{k+1} = \sum_{i=1}^s \frac{b_i}{d_0} (x_k + h d_0 f(t + c_i, y_i)),
        \end{align}
        and
        \begin{align}\label{pf1-2:ex}
            y_i =  x_k + h \sum_{j<i} a_{i,j} f(t+c_j, y_j).
        \end{align}    
    \end{subequations}
    We further rearrange~\eqref{pf1-2:ex}. Using~\eqref{pf1-2:ex} and~$d_i :=
    \sum_{j<i} a_{i,j}$, compute
    \begin{align*}
        &x_k + h d_i f(t+c_j, y_j)
        \nonumber\\
        &\quad = y_j + h d_i f(t+c_j, y_j) + (x_k - y_j)
        \nonumber\\
        &\quad = y_j + h d_i f(t+c_j, y_j) - h \sum_{l<j} a_{j,l} f(t+c_l, y_l),
    \end{align*}
    Then,~\eqref{pf1-2:ex} can be rewritten as
    \begin{align}\label{pf1-3:ex}
        y_i 
        &= \sum_{j<i} \frac{a_{i,j}}{d_i} (x_k + h d_i f(t+c_j, y_j))
        \nonumber\\
        &= \sum_{j<i} \frac{a_{i,j}}{d_i} (y_j + h d_i f(t+c_j, y_j))
        \nonumber\\
        &\quad -  \sum_{j<i} \sum_{l<j} \frac{ha_{i,j}a_{j,l}}{d_i} f(t+c_l, y_l).
\end{align}
Lipschitz constants of~\eqref{pf1-1:ex} and~\eqref{pf1-3:ex} are given by~\eqref{eq1:ex} and~\eqref{eq2:ex}, respectively.~\qed
\end{pf}

To apply Theorem~\ref{thm:ex}, we need to estimate a Lipschitz constant of the map~$\id + h d_i f$. Since this corresponds to the forward Euler discretization, we can find several estimations under strong infinitesimal contractivity~$\olp{f} \le -\lambda$, $\lambda \in \bR_{>0}$ and Lipschitzness~$\lp{f} \le \ell$, $\ell \in \bR_{>0}$, including~\cite[Lemma 19]{AD-SJ-AVP-FB:24} for a general norm. Also, the standard proof in~\cite[Section 6]{EKR-SB:16} for monotone operator theory can be applied in the Euclidean case:
\begin{align}\label{eq:lp_l2}
    \lp{\id + h d_i f} 
    \le \sqrt{1-2h d_i \lambda + (h d_i \ell)^2}.
\end{align}
Note that~$1-2h d_i \lambda + (h d_i \ell)^2 < 1$ if and only if $0<h d_i<2\lambda/\ell^2$.


\subsection{Examples}

We apply Theorem~\ref{thm:ex} with the Lipschitz estimate~\eqref{eq:lp_l2}
to explicit Runge-Kutta methods with the number of stages from~$s=1$
to~$s=5$. The coefficients are as follows:
\begin{itemize}
    \item Forward Euler ($s=1$): $A = 0$ and~$b=1$;

    \item Heun’s method ($s=2$):
        \begin{align*}
            A = \begin{bmatrix}0 & 0 \\ 1 & 0 \end{bmatrix}, 
            \quad
            b = \begin{bmatrix}1/2 \\ 1/2\end{bmatrix};
        \end{align*}

    \item Heun’s method ($s=3$):
        \begin{align*}
            A = \begin{bmatrix}
            0 & 0 & 0 \\ 
            1/3 & 0 & 0 \\ 
            0 & 2/3 & 0 
            \end{bmatrix}, 
            \quad
            b = \begin{bmatrix}1/4 \\ 0 \\ 3/4\end{bmatrix};
        \end{align*}

    \item Classical Runge–Kutta method ($s=4$):
        \begin{align*}
            A = \begin{bmatrix}
            0 & 0 & 0 & 0 \\ 
            1/2 & 0 & 0 & 0 \\ 
            0 & 1/2 & 0 & 0 \\ 
            0 & 0 & 1 & 0 
            \end{bmatrix}, 
            \quad
            b = \begin{bmatrix}
            1/6 \\ 1/3 \\ 1/3 \\ 1/6
            \end{bmatrix};
        \end{align*}

    \item Strong stability preserving Runge–Kutta method ($s=5$):
        \begin{align*}
            A = \begin{bmatrix}
            0 & 0 & 0 & 0 & 0 \\ 
            1/4 & 0 & 0 & 0 & 0 \\ 
            1/8 & 1/8 & 0 & 0 & 0 \\ 
            0 & 0 & 1/2 & 0 & 0 \\
            3/16 & -3/8 & 3/8 & 9/16 & 0 
            \end{bmatrix}, 
            \quad
            b = \begin{bmatrix}
            1/6 \\ 0 \\ 2/3 \\ 1/6 \\ 0
            \end{bmatrix}.
        \end{align*}
\end{itemize}
Figures~\ref{fig1:ERK} and~\ref{fig2:ERK} show the Lipschitz
constants~$\rho$ in~\eqref{eq1:ex} for contractivity rate~$\lambda=1,2$ and Lipschitz
constant~$\ell = 2$. In these two cases, each method preserves
contractivity for sufficiently small step size~$h \in \bR_{>0}$.

\begin{figure}[t]
    \includegraphics[width=1\linewidth]{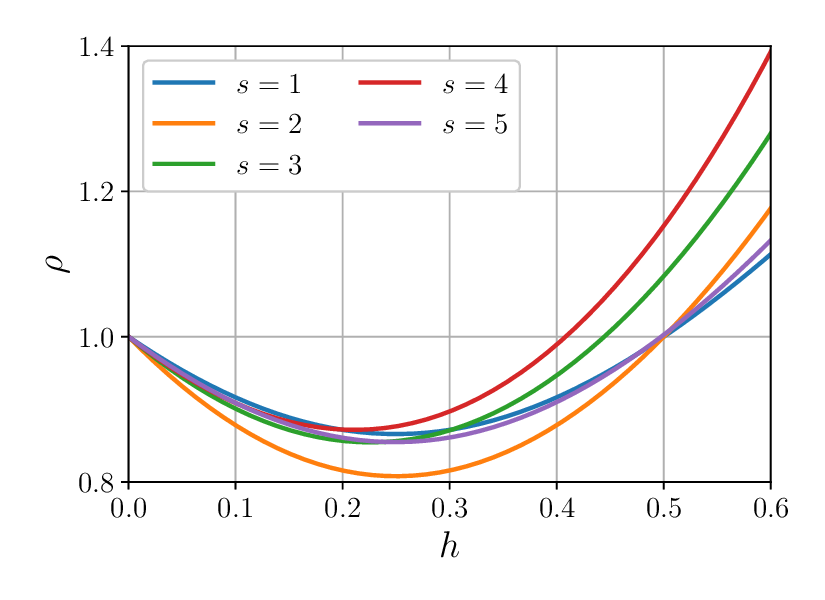}
    \caption{Explicit Runge-Kutta methods: Lipschitz constant~$\rho$ in~\eqref{eq1:ex} for~$\lambda = 1$ and~$\ell = 2$}
    \label{fig1:ERK}

    \includegraphics[width=1\linewidth]{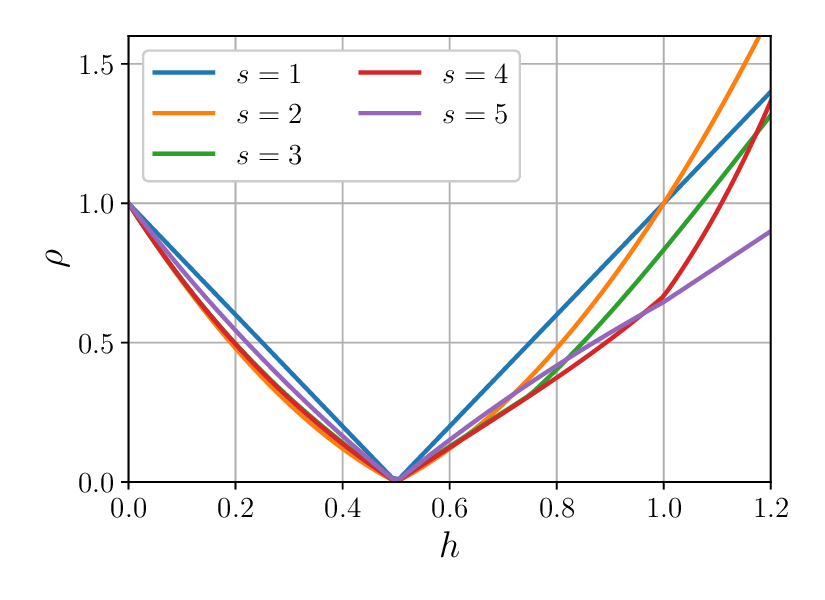}
    \caption{Explicit Runge-Kutta methods: Lipschitz constant~$\rho$ in~\eqref{eq1:ex} for~$\lambda = 2$ and~$\ell = 2$}
    \label{fig2:ERK}
\end{figure}



\section{Implicit Runge–Kutta Methods}\label{sec:imp}
In this section, we study contractivity preservation under implicit Runge–Kutta methods. In contrast to the explicit case, we also need to discuss the well definedness of implicit methods, i.e., the solvability of the implicit stage equation~\eqref{eq2:dsys}. 

\subsection{Well-Definedness}\label{sec:wd}
An~$s$-stage Runge--Kutta dynamics is said to be \emph{well-defined} if the implicit stage equation~\eqref{eq2:dsys} has a unique solution~$y^*=g(t_k,x_k)$ for each~$(t_k, x_k) \in \bR \times \bR^n$. A well-defined~$s$-stage Runge--Kutta dynamics admits a unique explicit form: 
\begin{align}\label{eq:esys}
    x_{k+1} = g (t_k, x_k), \quad k \in \bZ.
\end{align}
In the explicit case, an $s$-stage Runge--Kutta dynamics is always well-defined. In the implicit case, the well-definedness is a classical and extensively studied question; see, e.g.,~\cite[Theorem II.7.2]{HE-NSP-WG:93} and~\cite[Theorems IV.14.2 and IV.14.4]{HE-WG:96}. 

We show that these classical results can be viewed as special cases of a contractivity condition for the following \emph{auxiliary (continuous-time) dynamics} associated with~\eqref{eq2:dsys}, weighted by invertible~$Q \in \mathbb{R}^{sn \times sn}$:
\begin{align}\label{eq:asys}
    \dot y (t) = Q^{-1} \bigl( - y(t) + (\1_s \otimes x) + h (A \otimes \I_n) F_c (t, y(t)) \bigr),
\end{align}
where~$y(t) \in \bR^{sn}$ denotes the state.

\begin{thm}\label{thm:wd}
\textbf{\textup{(Well-Definedness Analysis via Auxiliary Dynamics)}}
    Let~$\| \cdot \|$ be a norm on~$\bR^{sn}$ with a compatible weak pairing, and let~$f:\bR \times \bR^n \to \bR^n$ be continuous. 
    Suppose that the auxiliary dynamics~\eqref{eq:asys} with invertible~$Q \in \bR^{sn \times sn}$ is strongly infinitesimally contracting with respect to norm~$\| \cdot \|$.
    Then    
    \begin{enumerate}
    
        \item \label{P1:wd} 
        there exists a unique~$G:\bR \times \bR^n \to \bR^{sn}$ such that
            \begin{align}\label{eq:alg}
                &- G(t, x) + (\1_s \otimes x) 
                \nonumber\\
                &\quad\qquad + h (A \otimes \I_n) F_c (t, G(t, x)) 
                = \0_{sn}
            \end{align}
        for all~$(t, x) \in \bR \times \bR^n$;
        
    
        \item \label{P2:wd} 
        $G (t, x) = \1_s  \otimes x$ if~$f(t, x) = \0_n$ and~$c=\0_s$;
        

        \item \label{P3:wd} 
        the explicit form~\eqref{eq:esys} is unique on~$\bR \times \bR^n$, given by
        \begin{align}\label{eq:g}
            g (t, x) = x + h (b \otimes \I_n)^\top F_c (t, G(t, x));
        \end{align}
        
    
        \item \label{P4:wd} 
        $y \mapsto \hat F_{h,A,c,t} (y) := y - h (A \otimes \I_n) F_c (t,y)$ is bijective on~$\bR^{sn}$ at each~$h \in \bR_{>0}$ and every~$A \in \bR^{s \times s}$,~$c \in \bR^s$ and~$t \in \bR$, and~$g (t, x)$ can be rewritten as
        \begin{subequations}\label{eq2:g}
            \begin{align}
                &g(t,x) 
                = \frac{1}{s} (\1_s \otimes \I_n)^\top y|_{y=\hat F_{h,A,c,t}^{-1}(\1_s \otimes x)}
                \nonumber\\
                &\qquad
                + h ( v \otimes \I_n )^\top F_c (t,y)|_{y=\hat F_{h,A,c,t}^{-1}(\1_s \otimes x)}\\
                &\quad
                v := b - \frac{1}{s} A^\top \1_s.
                \label{eq2:v}
        \end{align}
        \end{subequations}
    \end{enumerate}
\end{thm}
\begin{pf}
    The proof is in Appendix~\ref{app:thm:wd}.
    \qed
\end{pf}

The auxiliary system~\eqref{eq:asys} enables verification of the well-definedness of an implicit~$s$-stage Runge--Kutta method. Furthermore, the forward Euler discretization of~\eqref{eq:asys} can be used to implement the Runge--Kutta method, i.e., to solve the algebraic equation~\eqref{eq2:dsys} numerically; see Section~\ref{sec:impliment} below for details.

The proof of Property~\ref{P1:wd} in Theorem~\ref{thm:wd} is simple. Nevertheless, we can cover classical results~\cite[Theorem II.7.2]{HE-NSP-WG:93} and~\cite[Theorems IV.14.2 and IV.14.4]{HE-WG:96} as special cases.

\begin{cor}\label{cor1:wd}
\textbf{\textup{(Recovering~\cite[Theorem II.7.2]{HE-NSP-WG:93})}}
    Let~$\| \cdot \|_i$ be a component norm on~$\bR^n$ with a compatible weak pairing, and let~$\| \cdot \| := \|[\begin{matrix} \| \cdot \|_i & \cdots & \| \cdot \|_i \end{matrix}]^\top\|_1$ be the composite norm on~$\bR^{sn}$. Assume that~$f:\bR \times \bR^n \to \bR^n$ is continuous. 
    Then the auxiliary dynamics~\eqref{eq:asys} with~$Q = \I_{sn}$ is strongly infinitesimally contractive with rate~$\lambda \in \bR_{>0}$ with respect to the composit norm~$\| \cdot \|$ if
    \begin{align}\label{eq:HE-NSP-WG:93}
        -1 + h \| A\|_1 \lpi{f} \le -\lambda,
    \end{align}
    where~$\lpi{f}$ denotes the Lipschitz constant of~$f$ with respect the component norm~$\| \cdot \|_i$.
\end{cor}
\begin{pf}
    The proof is in Appendix~\ref{app:cor1:wd}.
    \qed
\end{pf}

\begin{cor}\label{cor2:wd}
\textbf{\textup{(Recovering~\cite[Theorems IV.14.2 and IV.14.4]{HE-WG:96})}}
    Let~$f:\bR \times \bR^n \to \bR^n$ be continuous, and let~$A \in \bR^{s \times s}$ be invertible. Also, let~$d \in \bR_{>0}^s$, and~$P \succ \0_{n \times n}$. 
    Then the auxiliary dynamics~\eqref{eq:asys} with~$Q = A \otimes \I_n$ is strongly infinitesimally contractive with rate~$\lambda \in \bR_{>0}$ with respect to norm~$\| \cdot \|_{2,([d] \otimes P)^{1/2}}$ if
    \begin{align}\label{eq:HE-WG:96}
        \mu_{2,[d]^{1/2}} ( - A^{-1} ) + h \; \olpP{f} \le -\lambda.
    \end{align}
\end{cor}
\begin{pf}
    The proof is in Appendix~\ref{app:cor2:wd}.
    \qed
\end{pf}

The conditions in~\cite[Theorems IV.14.2 and IV.14.4]{HE-WG:96} are described by using~$\alpha_{[d]} ( A^{-1} )$ that is the largest number~$\alpha \in \bR$ satisfying
\begin{align*}
    \llangle u, A^{-1} u \rrangle_{2,[d]^{1/2}} 
    \ge \alpha \| u \|_{2,[d]^{1/2}}^2, 
    \quad \forall u \in \bR^s.
\end{align*}
In fact, using~\cite[eq.(2.39)]{FB:26}, we can show~$\mu_{2,[d]^{1/2}} ( {-} A^{-1} \bigr) = -\alpha_{[d]} \bigl( A^{-1} \bigr)$ as follows:
\begin{align*}
    \mu_{2,[d]^{1/2}} ( - A^{-1} )
    &= \sup_{u \in \bR^s \setminus \{\0_s \}} 
    \frac{\llangle - A^{-1} u, u \rrangle_{2,[d]^{1/2}}}{\| u \|_{2,[d]^{1/2}}^2}
    \nonumber\\
    &= - \inf_{u \in \bR^s \setminus \{\0_s \}} 
    \frac{\llangle A^{-1} u, u \rrangle_{2,[d]^{1/2}}}{\| u \|_{2,[d]^{1/2}}^2}
    \nonumber\\
    &= - \alpha_{[d]} ( A^{-1} ).
\end{align*}
This provides the interpretation of~$\alpha_{[d]} \bigl( A^{-1} \bigr)$ in terms of the log norm~$\mu_{2,[d]^{1/2}} ( - A^{-1} )$. 


\subsection{Contractivity Preservation}\label{sec:cp}
In this subsection, our goal is to understand when implicit Runge--Kutta methods preserve the strong infinitesimal contractivity of a continuous-time system. For weak infinitesimal contractivity, this property is known as B-stability~\cite[Definition IV.12.2]{HE-WG:96}. Although the contractivity of implicit Runge--Kutta dynamics can be studied directly as done for weak infinitesimal contractivity with respect to the $\ell_2$-norm~\cite[Theorem IV.12.23]{HE-WG:96}, this does not immediately relate with the preservation of contractivity. Consequently, the preservation of weak infinitesimal contractivity with respect to the~$\ell_2$-norm has been studied by a different approach~\cite[Theorem IV.12.4]{HE-WG:96}. We extend this analysis to establish a condition for the preservation of strong infinitesimal contractivity with respect to the~$\ell_2$-norm as follows.

\begin{thm}\label{thm:con2}
\textbf{\textup{(Contractivity Preservation with respect to~$\ell_2$-Norms)}}
    Let~$f:\bR \times \bR^n \to \bR^n$ be continuous, and let~$P \succ \0_{n \times n}$. For a continuous-time system~\eqref{eq:sys} and its~$s$-stage Runge--Kutta dynamics~\eqref{eq:dsys}, assume
    \begin{enumerate}

        \item \label{A1:con2}
        \eqref{eq:dsys} is well-defined;

        
        \item \label{A2:con2}
        $\olpP{f} \le - \lambda_2$ for some~$\lambda_2 \in \bR_{>0}$;


        \item \label{A3:con2}
        $\lpP{f} \le \ell_2$ for some~$\ell_2 \in \bR_{>0}$;

        \item \label{A4:con2}
        $M:= [b] A + A^\top [b] - b b^\top \succeq \0_{s \times s}$ and~$b \ge \0_s$. 

    \end{enumerate}
    Then, the discrete-time system~\eqref{eq:dsys} is strongly contractive with factor~$\rho_2 \in [0, 1)$ with respect to norm~${\|\cdot\|_{2, P^{1/2}}}$:
    \begin{align}\label{eq:rho_con2}
        \rho_2 := \left( 1 - \frac{2 h \lambda_2 \| b\|_1 }{ \|\I_s + h \ell_2 | A | \|_{2, [b]^{1/2}}^2} \right)^{1/2}.
    \end{align}
\end{thm}

\begin{pf}
    The proof is in Appendix~\ref{app:thm:con2}.~\qed
\end{pf}

From~\eqref{eq:rho_con2}, an~$s$-stage Runge--Kutta dynamics preserves strong infinitesimal contractivity with respect to norm~${\|\cdot\|_{2, P^{1/2}}}$ for any step size~$h \in \bR_{>0}$ if all conditions in Theorem~\ref{thm:con2} hold. Note that any explicit Runge-Kutta method does not satisfy Assumption~\ref{A4:con2} because~$A$ is lower triangular with zeros on the diagonal.

The conditions in Assumption~\ref{A4:con2} of Theorem~\ref{thm:con2} is called algebraic stability~\cite[Definition 12.5]{HE-WG:96}. If we require the strict inequality~$b \in \bR_{>0}^s$, then Assumption~\ref{A4:con2} implies the well-definedness of an~$s$-stage Runge--Kutta method~\cite[Theorem 3.1]{SMN-85}. Thus, Assumption~\ref{A1:con2} can be removed.

We next investigate the preservation of strong infinitesimal contractivity with respect to the~$\ell_1$- and~$\ell_\infty$-norms.

\begin{thm}\label{thm:con1}
\textbf{\textup{(Contractivity Preservation with respect to~$\ell_1$-Norms)}}
    Let~$f:\bR \times \bR^n \to \bR^n$ be continuous, let~$\eta \in \bR_{> 0}^n$, and recall~$v = b - A^\top \1_s/s$ in~\eqref{eq2:v}. For a continuous-time system~\eqref{eq:sys} and its~$s$-stage Runge--Kutta dynamics~\eqref{eq:dsys}, assume
    
    \begin{enumerate}  
        \item \label{A1:con1}
        \eqref{eq:dsys} is well-defined;
        
        
        \item \label{A2:con1}
        $\olpeta{f} \le - \lambda_1$ for some~$\lambda_1 \in \bR_{>0}$;
        

        \item \label{A3:con1}
        $\lpeta{f} \le \ell_1$ for some~$\ell_1 \in \bR_{>0}$;
        

        \item \label{A4:con1}
        $a_{i,i} \ge 0$, $i=1,\dots,s$,~$v \ge \0_s$, and
        \begin{align}\label{eq:con1}
            \lambda_1 v^\top \1_s 
            > \max_{j \in \{1,\dots,s\} } 
            \biggl( - \lambda_1 a_{j,j} + \ell_1 \sum_{i \neq j} |a_{i,j}| \biggr);
        \end{align}
        

        \item \label{A5:con1}
        the step size~$h \in \bR_{>0}$ satisfies
        \begin{align}
            1 &\ge h \lambda_1 v^\top \1_s
            \label{eq1:h_con1}\\
            1 &\ge h s \hat \ell_i v_k, 
            \label{eq2:h_con1}
            \quad \forall i=1,\dots,n, \; \forall k=1,\dots,s,
        \end{align}
        where~$\hat \ell_i$,~$i=1,\dots,n$ denotes a Lipschitz constant of the~$i$th component~$f_i$ with respect to~$x_i$ uniformly in the other variables.
    \end{enumerate}
     Then,~\eqref{eq:dsys} is strongly infinitesimally contractive with factor~$\rho_1 \in [0, 1)$ with respect to norm~$\| \cdot \|_{1, [\eta]}$:
    \begin{align}\label{eq:rho_con1}
        \rho_1 
        := \frac{1- h \lambda_1 v^\top \1_s }{\displaystyle 1 - h \max_{j \in \{1,\dots,s\}} 
        \biggl( - \lambda_1 a_{j,j} + \ell_1 \sum_{i \neq j} |a_{i,j}| \biggr)}.
    \end{align}
\end{thm}

\begin{pf}
    The proof is in Appendix~\ref{app:thm:con1}.~\qed
\end{pf}

\begin{thm}\label{thm:coni}
\textbf{\textup{(Contractivity Preservation with respect to~$\ell_\infty$-Norms)}}
    Let~$f:\bR \times \bR^n \to \bR^n$ be continuous, let~$\eta \in \bR_{> 0}^n$, and recall~$v = b - A^\top \1_s/s$ in~\eqref{eq2:v}. For a continuous-time system~\eqref{eq:sys} and its~$s$-stage Runge--Kutta dynamics~\eqref{eq:dsys}, assume
    
    \begin{enumerate}
        \item \label{A1:coni}
        \eqref{eq:dsys} is well-defined;
        
        
        \item \label{A2:coni}
        $\olpetai{f} \le - \lambda_\infty$ for some~$\lambda_\infty \in \bR_{>0}$;
        

        \item \label{A3:coni}
        $\lpetai{f} \le \ell_\infty$ for some~$\ell_\infty \in \bR_{>0}$;
        

        \item \label{A4:coni}
        $a_{i,i} \ge 0$, $i=1,\dots,s$,~$v \ge \0_s$, and
        \begin{align*}
            \lambda_\infty v^\top \1_s 
            > \max_{i \in \{1,\dots,s\}} 
            \biggl( - \lambda_\infty a_{i,i} + \ell_\infty \sum_{j \neq i} |a_{j,i}| \biggr);
        \end{align*}
        

        \item \label{A5:coni}
        step size~$h \in \bR_{>0}$ satisfies~\eqref{eq1:h_con1} and~\eqref{eq2:h_con1}.
    \end{enumerate}
     Then,~\eqref{eq:dsys} is strongly infinitesimally contractive with factor~$\rho_\infty \in [0, 1)$ with respect to norm~$\| \cdot \|_{\infty, [\eta]^{-1}}$:
    \begin{align*}
        \rho_\infty 
        := \frac{1- h \lambda_\infty v^\top \1_s }{\displaystyle 1 - h \max_{i \in \{1,\dots,s\}} 
        \biggl( - \lambda_\infty a_{i,i} + \ell_\infty \sum_{j \neq i} |a_{i,j}| \biggr)}.
    \end{align*}
\end{thm}

\begin{pf}
    The proof is similar to that of Theorem~\ref{thm:con1}.~\qed
\end{pf}


\subsection{Implementation of Implicit Methods}\label{sec:impliment}
To implement an~$s$-stage implicit Runge--Kutta method, we need to solve the implicit stage equation~\eqref{eq2:dsys} with respect to~$y$. As shown in Property~\ref{P1:wd} of Theorem~\ref{thm:wd}, the algebraic equation has a unique solution~$y^*=g(t_k,x_k)$ under the strong infinitesimal contractivity of the auxiliary continuous-time dynamics~\eqref{eq:asys}. Moreover,~$y^*=g(t_k,x_k)$ is the equilibrium of the auxiliary dynamics. If we can compute~$y^*=g(t_k,x_k)$, then the~$s$-stage Runge--Kutta method can be implemented.

Since the forward Euler-discretization preserves an equilibrium of a continuous-time dynamics, we can use the forward Euler-discretization of the auxiliary dynamics~\eqref{eq:asys} for computing~$y^*=g(t_k,x_k)$. According to~\cite[Lemma 19]{AD-SJ-AVP-FB:24}, a Lipschitz constant of the forward Euler-discretization can be estimated by using contractivity rate and Lipschitz constant of a continuous-time dynamics and step size~$h \in \bR_{>0}$ for discretization. Applying this estimation to the auxiliary dynamics~\eqref{eq:asys}, its forward Euler-discretization is strongly infinitesimally contracting with factor~$\rho \in [0,1)$ if step size~$h$ is selected sufficiently small as to satisfy
\begin{align}\label{eq:step}
    \rho := \e^{-h \lambda} + \e^{h \ell} - 1 - h \ell < 1,
\end{align}
where~$\lambda$ and~$\ell$ denote a one-sided Lipschitz constant and Lipschitz constant of~$Q^{-1}(- \id_{sn} + h (A \otimes \I_n) F_c)$, respectively. Therefore, if the auxiliary dynamics~\eqref{eq:asys} is strongly infinitesimally contracting and has a finite Lipschitz constant, then its forward Euler-discretization with step size~$h$ satisfying~\eqref{eq:step} can be used to implement an~$s$-stage Runge--Kutta method.



\subsection{Numerical Examples}

\subsubsection{Implicit Midpoint Method}
As an example of implicit Runge--Kutta methods, we consider the implicit midpoint method whose coefficients are~$a_{1,1} =  1/2$ and~$b_1 = 1$~\cite[p.205]{HE-NSP-WG:93}. A one-sided Lipschitz constant of the auxiliary dynamics~\eqref{eq:asys} with~$Q=\I_n$ is~$-1 + (h/2) \olp {f}$. According to Theorem~\ref{thm:wd}, the implicit midpoint dynamics is well-defined for any~$h \in \bR_{>0}$ such that~$-1 + (h/2) \olp {f} < 0$. In particular, if the original continuous-time dynamics is strongly infinitesimally contracting, then the implicit midpoint dynamics is well-defined for any~$h \in \bR_{>0}$.

Next, we show contractivity preservation. For the~$\ell_2$-norm, Assumptions~\ref{A1:con2} and~\ref{A4:con2} in Theorem~\ref{thm:con2} hold. If the original continuous-time dynamics is strongly infinitesimally contracting with rate~$\lambda_2 \in \bR_{>0}$ and has a finite Lipschitz constant~$\ell_2 \in \bR_{>0}$ with respect to norm~${\|\cdot\|_{2, P^{1/2}}}$, then from~\eqref{eq:rho_con2}, a contractivity factor of the implicit midpoint dynamics with respect to norm~${\|\cdot\|_{2, P^{1/2}}}$ is
\begin{align*}
    \rho_2 = \left( 1 - \frac{2 h \lambda_2}{(1 + h \ell_2/2)^2} \right)^{1/2} < 1.
\end{align*}
This is minimized as~$\rho_2 = (1 - \lambda_2/\ell_2)^{1/2}$ by selecting~$h = 2/\ell_2$.

For the~$\ell_1$-norm, Assumption~\ref{A1:con1} in Theorem~\ref{thm:con2} holds, where~$v=1/2$. Also, if the original continuous-time dynamics is strongly infinitesimally contracting with rate~$\lambda_1 \in \bR_{>0}$ and has a finite Lipschitz constant~$\ell_1 \in \bR_{>0}$ with respect to norm~$\| \cdot \|_{1, [\eta]}$, then Assumptions~\ref{A2:con1} -- \ref{A4:con1} hold. Finally, since~$\ell_1 \ge \lambda_1$ and~$\ell_1 \ge \hat \ell_i$,~$i=1,\dots,n$, Assumption~\ref{A5:con1} holds if~$1 \ge h \ell_1/2$. In this case, a contractivity factor of the implicit midpoint dynamics with respect to norm~$\| \cdot \|_{1, [\eta]}$ is
\begin{align*}
    \rho_1 = \frac{1- h \lambda_1/2}{1 + h \lambda_1/2} < 1.
\end{align*}
This is a decreasing function of~$h$. From~$1 \ge h \ell_1/2$, $h=2/\ell_1$ gives the smallest contractivity factor~$\rho_1 = \frac{1- \lambda_1/\ell_1}{1 + \lambda_1/\ell_1}$. A similar estimation is obtained for the~$\ell_\infty$-norm.


\subsubsection{Implicit Euler Method}
For the implicit Euler method, the coefficients are~$a_{1,1} = b_1 = 1$. A one-sided Lipschitz constant of the auxiliary dynamics~\eqref{eq:asys} with~$Q=\I_n$ is~$-1 + h \olp {f}$. Thus, for a strongly infinitesimally contracting continuous-time system, the implicit Euler dynamics is well-defined for any~$h \in \bR_{>0}$. 

Next, we show contractivity preservation. Suppose that the original continuous-time dynamics is strongly infinitesimally contracting with rate~$\lambda_2 \in \bR_{>0}$ and has a finite Lipschitz constant~$\ell_2 \in \bR_{>0}$ with respect to norm~${\|\cdot\|_{2, P^{1/2}}}$. Then, all assumptions in Theorem~\ref{thm:con2} hold, and a contractivity factor of the implicit Euler dynamics  with respect to norm~${\|\cdot\|_{2, P^{1/2}}}$ is
\begin{align*}
    \rho_2 = \left( 1 - \frac{2 h \lambda_2}{ (1 + h \ell_2 )^2} \right)^{1/2} < 1.
\end{align*}
This is minimized as~$\rho_2 = ( 1 - \lambda_2/2\ell_2 )^{1/2}$ by selecting~$h=1/\ell_2$.

Suppose that the original continuous-time dynamics is strongly infinitesimally contracting with rate~$\lambda_1 \in \bR_{>0}$ and has a finite Lipschitz constant~$\ell_1 \in \bR_{>0}$ with respect to norm~$\| \cdot \|_{1, [\eta]}$. Since~$v=0$, all assumptions in Theorem~\ref{thm:con1} hold, and a contractivity factor of the implicit Euler dynamics  with respect to norm~$\| \cdot \|_{1, [\eta]}$ is 
\begin{align*}
    \rho_1 = \frac{1}{1 + h \lambda_1}.
\end{align*}
This can be made arbitrary small by selecting~$h \in \bR_{>0}$ sufficiently large. A similar estimation is obtained for the~$\ell_\infty$-norm.



\section{Conclusion}\label{sec:con}
In this paper, we have revisited classical problems of multi-stage Runge--Kutta methods using tools originally developed for their analysis and further refined through contraction theory. We have first estimated Lipschitz constants of explicit multi-stage Runge--Kutta dynamics. We have next shown that implicit multi-stage Runge--Kutta dynamics are well-defined when the corresponding auxiliary dynamics are strongly infinitesimally contractive, which covers classical well-definedness conditions as special cases. Moreover, the auxiliary dynamics have provided a means to implement multi-stage Runge--Kutta methods without directly solving the associated implicit equations. We have then studied conditions under which multi-stage Runge--Kutta methods preserve strong infinitesimal contractivity of the original continuous-time dynamics with respect to the~$\ell_2$-, $\ell_1$-, and~$\ell_\infty$-norms. Our condition for the~$\ell_2$-norm has been based on the classical analysis for preservation of weak infinitesimal contractivity, while such an analysis has not been conducted for the~$\ell_1$- and~$\ell_\infty$-norms before.



\begin{ack}                               
This work of Y. Kawano is supported in part by JST FOREST Program Grant Number JPMJFR222E and JSPS KAKENHI Grant Number JP21H04875. This work of F. Bullo is supported in part by AFOSR grant FA9550-22-1-0059. 
\end{ack}



\appendix


\section{Proofs of Results in Section~\ref{sec:wd}}
\subsection{Proof of Theorem~\ref{thm:wd}}\label{app:thm:wd}
    (Proof of~\ref{P1:wd})
    At each~$(t, x) \in \bR \times \bR^n$, there is a one-to-one correspondence between a solution to the algebraic equation~\eqref{eq2:dsys} and an equilibrium of the auxiliary dynamics~\eqref{eq:asys}, where recall that~$Q$ is invertible. The contractivity of~\eqref{eq:asys} implies the uniqueness of its equilibrium~$y^* = G(t, x)$ at each~$(t, x) \in \bR \times \bR^n$~\cite[Theorem 3.9]{FB:26}.


    (Proof of~\ref{P2:wd})
    If~$f(t, x) = \0_n$ and~$c = \0_s$, then~$y = \1_s  \otimes x$ satisfies~\eqref{eq2:dsys}. From Property~\ref{P1:wd}, i.e., the uniqueness of~$G(t,x)$, we have~$G(t, x) = \1_s  \otimes x$.
    

    (Proof of~\ref{P3:wd})
    From Property~\ref{P1:wd},~$y = G(t,x)$ is a unique solution to~\eqref{eq2:dsys}. Substituting~$y = G(t,x)$ into~\eqref{eq1:dsys} yields~\eqref{eq:g}. 
    

    (Proof of~\ref{P4:wd})
    The contractivity of~\eqref{eq:asys} implies that~\eqref{eq2:dsys} has a unique solution when~$x=\0_n$. Thus,~$\hat F_{h,A,c,t}(y) = y - h (A \otimes \I_n) F_c (t,y)$ is bijective. 

    The bijectivity of~$\hat F_{h,A,c,t}(y)$ implies that~\eqref{eq2:dsys} can be rewritten as~$\1_s \otimes x_k = \hat F_{h,A,c,t}(y)$, i.e.,~$y = \hat F_{h,A,c,t}^{-1} (\1_s \otimes x_k)$. Thus,~\eqref{eq:dsys} is equivalent to
        \begin{align}\label{pf1:P3:wd}
            x_{k+1} 
            = x_k + h (b \otimes \I_n)^\top F_c ( t_k, y )|_{y = \hat F_{h,A,c,t}^{-1} (\1_s \otimes x_k)}.
        \end{align}
    Also, using the transpose and mixed product properties of the Kronecker product~\cite[E4.6]{FB:26},~\eqref{eq2:dsys}, and~$\hat F_{h,A,c,t}^{-1}(y) = y - h (A \otimes \I_n) F_c (t,y)$, we can rewrite~$x_k$ as follows
        \begin{align*}
            x_k & = \frac{1}{s} (\1_s \otimes \I_n)^\top (\1_s \otimes x_k) \\
            &= \frac{1}{s} (\1_s \otimes \I_n)^\top \hat F_{h,A,c,t_k}(y)|_{y = \hat F_{h,A,c,t_k}^{-1} (\1_s \otimes x_k)} \\
            &= \frac{1}{s} (\1_s \otimes \I_n)^\top y|_{y = \hat F_{h,A,c,t_k}^{-1} (\1_s \otimes x_k)}
            \nonumber\\
            &\quad 
            - \frac{h}{s} (\1_s \otimes \I_n)^\top (A \otimes \I_n) F_c (t_k,y)|_{y = \hat F_{h,A,c,t_k}^{-1} (\1_s \otimes x_k)}.
        \end{align*}
    Substituting this into~\eqref{pf1:P3:wd} yields
        \begin{align*}
            x_{k+1} 
            &=  \frac{1}{s} (\1_s \otimes \I_n)^\top y|_{y = \hat F_{h,A,c,t_k}^{-1} (\1_s \otimes x_k)}
            \nonumber\\
            &\quad + h \left( (b \otimes \I_n)^\top- \frac{1}{s} (\1_s \otimes \I_n)^\top (A \otimes \I_n) \right)
            &\nonumber\\
            &\qquad \times F_c (t_k,y)|_{y = \hat F_{h,A,c,t_k}^{-1} (\1_s \otimes x_k)}.
        \end{align*}
    From the the transpose and bilinearity properties of the Kronecker product~\cite[E4.6]{FB:26},~\eqref{eq:esys} holds for~$g(t,x)$ in~\eqref{eq2:g}.~\qed


\subsection{Proof of Corollary~\ref{cor1:wd}}\label{app:cor1:wd}
    From~the translation and positive homogeneity properties of the one-sided Lipschitz constant~\cite[Lemma 3.5]{FB:26}, we have
    \begin{align*}
        &\olp{- \id_{sn} + (\1_s \otimes x) + h (A \otimes \I_n) F_c} 
        \nonumber\\
        &\quad =  -1 +  h \; \olp{(A \otimes \I_n) F_c}.
    \end{align*}
    From the relation between Lipschitz and one-sided Lipschitz constants~\cite[eq. (3.15)]{FB:26} and~\cite[Definition 3.2]{FB:26}, we obtain
    \begin{align*}
        \olp{(A \otimes \I_n) F_c}
        &\le \lp{(A \otimes \I_n) F_c}
        \\
        &\le  \| A \otimes \I_n \| \lp{F_c}.
    \end{align*}
    From~$\| \cdot \|=\|[\begin{matrix} \| \cdot \|_i & \cdots & \| \cdot \|_i \end{matrix}]^\top\|_1$, it follows that
    \begin{align*}
        \| A \otimes \I_n \| = \| A \|_1,
        \quad
        \lp{F_c} \le \lpi{f}.
    \end{align*}
    Combining these all lead to
    \begin{align*}
        &\olp{- \id_{sn} + (\1_s \otimes x) + h (A \otimes \I_n) F_c} 
        \nonumber\\
        &\quad \le -1 +  h \| A \|_1 \lpi{f}.
    \end{align*}
    Thus,~\eqref{eq:HE-NSP-WG:93} implies the strong infinitesimal contractivity of the auxiliary dynamics~\eqref{eq:asys} with~$Q = \I_{sn}$ with rate~$\lambda \in \bR_{>0}$.~\qed


\subsection{Proof of Corollary~\ref{cor2:wd}}\label{app:cor2:wd}
    From the subadditivity and positive homogeneity of the one-sided Lipschitz constant~\cite[Lemma 3.5]{FB:26}, and the relation between one-sided Lipschitz constant and log norm~\cite[Lemma 3.4]{FB:26}, we have
    \begin{align}\label{pf1:HE-WG:96}
        &\olp{(A \otimes \I_n)^{-1}(- \id_{sn} + (\1_s \otimes x) + h (A \otimes \I_n) F_c)} 
        \nonumber\\
        &\quad \le \mu (- (A \otimes \I_n)^{-1} ) + h \; \olp{F_c}.
    \end{align}
    We estimate an upper bound on each term of the right-hand side by selecting~$\| \cdot \| = \| \cdot \|_{2,([d] \otimes P)^{1/2}}$.

    Using the inverse, mixed product, transpose, and bilinearity properties of the Kronecker product~\cite[E4.6]{FB:26}, compute
    \begin{align}\label{pf2:HE-WG:96}
        &([d] \otimes P) \bigl( - ( A \otimes \I_n)^{-1} \bigr) 
        + \bigl( - ( A \otimes \I_n)^{-1} \bigr)^\top ([d] \otimes P)
        \nonumber\\
        &\quad =([d] \otimes P) \bigl( (- A^{-1}) \otimes \I_n \bigr)
        \nonumber\\
        &\qquad 
        + \bigl( (- A^{-1}) \otimes \I_n \bigr)^\top ([d] \otimes P)
        \nonumber\\
        &\quad = ([d] (- A^{-1})) \otimes P + ((- A^{-1})^\top [d]) \otimes P
        \nonumber\\
        &\quad = \bigl( [d]  (- A^{-1}) + (- A^{-1})^\top [d]  \bigr) \otimes P.
    \end{align}
    On the other hand, from the definition of~$\mu_{2,[d]^{1/2}} (- A^{-1})$, it follows that
    \begin{align*}
        [d]  (- A^{-1}) + (- A^{-1})^\top [d]
        \preceq \mu_{2,[d]^{1/2}} (- A^{-1}) [d].
    \end{align*}
    From the eigenvalue property of the Kronecker product~\cite[E4.6]{FB:26} and~$P \succ \0_{n \times n}$, we have
    \begin{align*}
        &\bigl( [d]  (- A^{-1}) + (- A^{-1})^\top [d] \bigr) \otimes P
        \nonumber\\
        &\quad \preceq (\mu_{2,[d]^{1/2}} (- A^{-1}) [d]) \otimes P
        \nonumber\\
        &\quad = \mu_{2,[d]^{1/2}} (- A^{-1}) ([d] \otimes P).
    \end{align*}
    This and~\eqref{pf2:HE-WG:96} yield
    \begin{align*}
        &([d] \otimes P) \bigl( - ( A \otimes \I_n)^{-1} \bigr) 
        + \bigl( - ( A \otimes \I_n)^{-1} \bigr)^\top ([d] \otimes P)
        \nonumber\\
        &\quad \preceq \mu_{2,[d]^{1/2}} (- A^{-1}) ([d] \otimes P).
    \end{align*}
    From the definition of~$\mu_{2,([d] \otimes P)^{1/2}} ( - ( A \otimes \I_n)^{-1} )$, we have
    \begin{align}\label{pf3:HE-WG:96}
        \mu_{2,([d] \otimes P)^{1/2}} \left( - ( A \otimes \I_n)^{-1} \right)
        \le \mu_{2,[d]^{1/2}} (- A^{-1}).
    \end{align}

    Next, since~$[d] \succ \0_{s \times s}$ is diagonal, it follows from~\eqref{eq3:dsys} that
    \begin{align*}
        &\llangle F_c(t,y) - F_c(t,y'), y- y' \rrangle_{2, ([d] \otimes P)^{1/2}} \\
        &\quad = \sum_{i=1} d_i \llangle f(t+c_i, y_i) - f(t+c_i, y'_i), y_i - y'_i \rrangle_{2, P^{1/2}} \\
        &\quad \le \olpP{f} \sum_{i=1} d_i \| y_i - y'_i \|_{2, P^{1/2}}^2 \\
        &\quad  = \olpP{f} \| y - y' \|_{2, ([d] \otimes P)^{1/2}}^2.
    \end{align*}
    From the definition of~$\olpDP{F_c}$, we obtain
    \begin{align}\label{pf4:HE-WG:96}
        \olpDP{F_c} \le \olpP{f}.
    \end{align}

    From~\eqref{pf1:HE-WG:96},~\eqref{pf3:HE-WG:96}, and~\eqref{pf4:HE-WG:96} with~$\| \cdot \| = \| \cdot \|_{2,([d] \otimes P)^{1/2}}$, we have
    \begin{align*}
        &\olp{(A \otimes \I_n)^{-1}(- \id_{sn} + (\1_s \otimes x) + h (A \otimes \I_n) F_c)} 
        \nonumber\\
        &\quad \le \mu_{2,[d]^{1/2}} (- A^{-1}) + h \; \olpP{f}
    \end{align*}
    Thus,~\eqref{eq:HE-WG:96} implies the strong infinitesimal contractivity of the auxiliary dynamics~\eqref{eq:asys} with~$Q = A \otimes \I_n$ with rate~$\lambda$.~\qed



\section{Proofs of Results in Section~\ref{sec:cp}}
\subsection{Proof of Theorem~\ref{thm:con2}}\label{app:thm:con2}
    From~\eqref{eq2:dsys} and~\eqref{eq3:dsys}, i.e.,
    \begin{align*}
        y_i = x_k + h \sum_{j=1}^s a_{i,j} f (t_k + c_j, y_j),
    \end{align*}
    the triangular inequality of the vector norm, and Assumptions~\ref{A3:con2}, we have
    \begin{align*}
        &\|y_i - y'_i \|_{2, P^{1/2}} \\
        &\quad \ge \|x_k - x_k'\|_{2, P^{1/2}} \\
        &\qquad 
        - h \left\| \sum_{j=1}^s a_{i,j} (f (t_k+c_j, y_j) - f (t_k+c_j, y'_j)) \right\|_{2, P^{1/2}} \\
        &\quad \ge \|x_k - x_k'\|_{2, P^{1/2}} - h \ell_2 \sum_{j=1}^s |a_{i,j}| \|y_j - y'_j\|_{2, P^{1/2}},
    \end{align*}
    i.e.,
    \begin{align*}
        \|x_k - x_k'\|_{2, P^{1/2}} \1_s \le ( \I_s + h \ell_2 | A | ) 
        \begin{bmatrix}
        \|y_1 - y'_1\|_{2, P^{1/2}} \\ \vdots \\ \|y_s - y'_s\|_{2, P^{1/2}}
        \end{bmatrix}.
    \end{align*}
    Taking~$\| \cdot \|_{2, [b]^{1/2}}$ (that is not necessary a norm) preserves the order of the inequality because of~$[b] \succeq \0_{s\times s}$~in Assumption~\ref{A4:con2}. Thus, it follows that
    \begin{align*}
        &\| b\|_1^{1/2} \|x_k - x_k'\|_{2, P^{1/2}} 
        \\
        &\quad \le 
        \|\I_s + h \ell_2 | A | \|_{2, [b]^{1/2}}
        \left\|\begin{bmatrix}
        \|y_1 - y'_1\|_{2, P^{1/2}} \\ \vdots \\ \|y_s - y'_s\|_{2, P^{1/2}}
        \end{bmatrix} \right\|_{2, [b]^{1/2}}.
    \end{align*}
    Squaring both sides yields
    \begin{align}\label{pf1:con2}
        &\frac{\| b\|_1 \|x_k - x_k'\|_{2, P^{1/2}}^2}{\|\I_s + h \ell_2 | A | \|_{2, [b]^{1/2}}^2} 
        \le 
        \left\|\begin{bmatrix}
        \|y_1 - y'_1\|_{2, P^{1/2}} \\ \vdots \\ \|y_s - y'_s\|_{2, P^{1/2}}
        \end{bmatrix} \right\|_{2, [b]^{1/2}}^2.
    \end{align}
    
    Next, according to the proof of~\cite[Theorem IV.12.4]{HE-WG:96}, under Assumptions~\ref{A2:con2} and~\ref{A4:con2}, the~$s$-stage Runge--Kutta dynamics~\eqref{eq:dsys} satisfies
    \begin{align}\label{pf2:con2}
        &\| x_{k+1} - x'_{k+1}\|_{2, P^{1/2}}^2
        \nonumber\\
        &\quad \le  \| x_k - x'_k \|_{2, P^{1/2}}^2
        - 2 h \lambda_2
        \left\|\begin{bmatrix}
        \|y_1 - y'_1\|_{2, P^{1/2}} \\ \vdots \\ \|y_s - y'_s\|_{2, P^{1/2}}
        \end{bmatrix} \right\|_{2, [b]^{1/2}}^2.
    \end{align}
    From~\eqref{pf1:con2} and~\eqref{pf2:con2},~\eqref{eq:dsys} is strongly infinitesimally contracting with factor~$\rho_2$ in~\eqref{eq:rho_con2}.~\qed


\subsection{Proof of Theorem~\ref{thm:con1}}\label{app:thm:con1}
From Property~\ref{P4:wd} of Theorem~\ref{thm:wd},~$g(t,x)$ can be described as in~\eqref{eq2:g} under Assumption~\ref{A1:con1}. We compute an upper bound on~$\lpeta{g}$. To this end, we first estimate an upper bound on~$\lpetaa{\hat F_{h,A,c,t}^{-1}}$. Using~$\hat F_{h,A,c,t} (y) = - h (A \otimes \I_n) F_c (t,y) + y$ in Property~\ref{P4:wd} of Theorem~\ref{thm:wd}, the definition of one-sided Lipschitz maps~\cite[Definition 3.2]{FB:26}, sub-additivity of a weak paring~\cite[Definition 2.27]{FB:26}, \cite[E2.27]{FB:26}, and the definition of the sign paring~\cite[Definition 2.29]{FB:26}, compute
    \begin{align}\label{pf1:con1}
        &\olpetaa{-\hat F_{h,A,c,t}} 
        \nonumber\\
        &\quad = \sup_{y\neq y'}\frac{\llbracket - \hat F_{h,A,c,t} (y) + \hat F_{h,A,c,t} (y'), y - y' \rrbracket_{1,[\1_s \otimes \eta]}}
        {\| y - y'\|_{1,[\1_s \otimes \eta]}^2}
        \nonumber\\
        &\quad = h \sup_{y\neq y'}\frac{\llbracket  (A \otimes \I_n) (F_c (t,y) - F_c (t,y')), y - y' \rrbracket_{1,[\1_s \otimes \eta]}}
        {\| y - y'\|_{1,[\1_s \otimes \eta]}^2}
        \nonumber\\
        &\qquad + \sup_{y\neq y'}\frac{\llbracket - (y - y'), y - y' \rrbracket_{1,[\1_s \otimes \eta]}}
        {\| y - y'\|_{1,[\1_s \otimes \eta]}^2}
        \nonumber\\
        &\quad = h \; \olpetaa{(A \otimes \I_n) F_c} - 1.
    \end{align}
    From Assumptions~\ref{A2:con1}~and~\ref{A3:con1},~$a_{i,i} \ge 0$,~$i=1,\dots,s$ in Assumption~\ref{A4:con1} and
    \begin{align*}
        &(A \otimes \I_n) F_c(t,y)
        \nonumber\\
        &\quad =\begin{bmatrix}
        a_{1,1} D_{y_1} f(t+c_1, y_1) & \cdots & a_{1,s} D_{y_s} f(t+c_s, y_s) \\
        \vdots & \ddots & \vdots \\
        a_{s,1} D_{y_1} f(t+c_1, y_1) & \cdots & a_{s,s} D_{y_s} f(t+c_s, y_s) \\
    \end{bmatrix},
    \end{align*}
    we obtain
    \begin{align}\label{pf2:con1}
        &\olpetaa{(A \otimes \I_n) F_c}
        \nonumber\\
        &\quad \le \max_{j \in \{1,\dots,s\}} 
        \biggl( - \lambda_1 a_{j,j} + \ell_1 \sum_{i \neq j} |a_{i,j}| \biggr).
    \end{align}
    Combining~\eqref{pf1:con1} and~\eqref{pf2:con1} leads to
    \begin{align*}
        &-\olpetaa{-\hat F_{h,A,c,t}}
        \nonumber\\
        &\quad \ge 
        1 - h \max_{j \in \{1,\dots,s\}} 
        \biggl( - \lambda_1 a_{j,j} + \ell_1 \sum_{i \neq j} |a_{i,j}| \biggr).
    \end{align*}
    Also,~\eqref{eq:con1} in Assumption~\ref{A4:con1} and~\eqref{eq1:h_con1} in Assumption~\ref{A5:con1} imply
    \begin{align*}
        1 - h \max_{j \in \{1,\dots,s\}} 
        \biggl( - \lambda_1 a_{j,j} + \ell_1 \sum_{i \neq j} |a_{i,j}| \biggr) > 0.
    \end{align*}    
    Therefore, from~\cite[Corollary 18]{AD-SJ-AVP-FB:24}, we have
    \begin{align}\label{pf3:con1}
        &\lpetaa{\hat F_{h,A,c,t}^{-1}} 
        \nonumber\\
        &\quad \le \frac{1}{\displaystyle 1 - h \max_{j \in \{1,\dots,s\}} 
        \biggl( - \lambda_1 a_{j,j} + \ell_1 \sum_{i \neq j} |a_{i,j}| \biggr)}.
    \end{align}

    Next, from Assumption~\ref{A2:con1},~$v \ge \0_s$ in Assumption~\ref{A4:con1}, and~\eqref{eq2:h_con1} in Assumption~\ref{A5:con1}, it follows that
    \begin{align*}
        &\eta^\top \left| \frac{1}{s} (y - y')  + h v_k (f(t+c_k, y_k) - f(t+c_k, y_k')) \right|
        \nonumber\\
        &\quad \le \left(\frac{1}{s} - h \lambda_1 v_k\right) \eta^\top |y - y'|,
        \quad k = 1,\dots,s,
    \end{align*}
    and consequently, from~\eqref{eq3:dsys},
    \begin{align}\label{pf4:con1}
        &\eta^\top \biggl| \frac{1}{s} (\1_s \otimes \I_n)^\top (y - y') 
        \nonumber\\
        &\qquad
        + h ( v \otimes \I_n )^\top (F_c (t,y) - F_c (t,y')) \biggr|
        \nonumber\\
        &\quad  \le \left(\frac{1}{s} \1_s - h \lambda_1 v \right)^\top \otimes \eta^\top | y - y'|.
    \end{align}

    Finally, from~\eqref{eq2:g},~\eqref{pf3:con1}, and~\eqref{pf4:con1}, we have~$\lpeta{g} \le \rho_1$ for~$\rho_1$ in~\eqref{eq:rho_con1}. Moreover,~$\rho_1 \in [0, 1)$ if~\eqref{eq:con1} in Assumption~\ref{A4:con1} holds.~\qed




\bibliographystyle{plain}        
\bibliography{autosam}           

\end{document}